\documentclass[reqno, a4paper]{amsart}
\usepackage{amsmath}
\usepackage{amssymb}
\usepackage{amsthm}

\usepackage[scale=0.9]{geometry}

\usepackage{natbib}
\usepackage{bibentry} % inline refereces

\usepackage[english]{babel}
\usepackage[utf8]{inputenc}

\usepackage{subfig}
\usepackage{graphicx}

\usepackage[unicode]{hyperref}
\usepackage[active]{srcltx} % use TeX-souce-specials-mode

\usepackage[bbgreekl]{mathbbol}
\usepackage{bm} % sophisticated \boldsymbol
\usepackage{MnSymbol} % \lsem, \rsem, tensor product :
\usepackage{gensymb}
\usepackage{eurosym}

\usepackage{units}
\usepackage{tensor}
\usepackage{accents}

\usepackage{enumitem}

\usepackage{lineno}

\usepackage{multicol}

\usepackage{comment}

\usepackage{placeins}

\usepackage{todonotes}

\DeclareMathOperator{\Tr}{Tr}

\DeclareMathOperator{\ad}{ad} % adjoint ad_X (Y) = [X,Y]  
\newcommand{\exponential}[1]{\ensuremath{{\mathrm e}^{#1}}}

\newcommand{\bydefinition}{\mathrm{def}}

\newcommand{\diff}{\mathrm{d}}
\renewcommand{\vec}[1]{\ensuremath{\mathbf{#1}}}

\makeatletter
\@ifpackageloaded{bm}% 
{\renewcommand{\vec}[1]{\ensuremath{\bm{#1}}}%
}{%
\relax% do nothing
}
\makeatother

\newcommand{\tensorq}[1]{\ensuremath{\mathbb{#1}}}      % tensorial quantity
\newcommand{\transpose}[1]{#1^\top}

\newcommand{\inverse}[1]{#1^{-1}}

\newcommand{\identity}{\ensuremath{\tensorq{I}}} % identity
\newcommand{\foidentity}{\ensuremath{\mathcal{I}}} % fourth order identity tensor
\newcommand{\tensorzero}{{\mathbb{O}}} % zero tensor

\newcommand{\fgrad}{\tensorq{F}}

\newcommand{\lcg}{\tensorq{B}}

\makeatletter
\@ifpackageloaded{bm}% 
{%
\newcommand{\linstrain}{\bbespilon} %requires \usepackage[bbgreekl]{mathbbol}
}{%
\newcommand{\linstrain}{\bbespilon} %requires \usepackage[bbgreekl]{mathbbol}
}

\@ifpackageloaded{bm}%
{%
 
}{%

}

\@ifpackageloaded{bm}%
{%
}{%
}
\makeatother

\newcommand{\henckystrain}{\tensorq{H}} % Hencky strain
\newcommand{\generictensor}{{\tensorq{A}}}
\newcommand{\gradasym}{\ensuremath{\tensorq{W}}}
\newcommand{\gradsym}{\ensuremath{\tensorq{D}}}

\newcommand{\gradvl}{\ensuremath{\tensorq{L}}}

\newcommand{\logspin}{\genspin[\mathrm{log}]}
\newcommand{\genspin}[1][\star]{\ensuremath{\tensorq{\Omega}}^{#1}}

\newcommand{\R}{\ensuremath{{\mathbb R}}}

\newcommand{\N}{\ensuremath{{\mathbb N}}}
\newcommand{\Z}{\ensuremath{{\mathbb Z}}}

\newcommand{\pd}[2]{\ensuremath{\frac{\partial {#1}}{\partial {#2}}}}

\newcommand{\dd}[2]{\ensuremath{\frac{\diff {#1}}{\diff {#2}}}}

\newcommand{\fid}[1]{\ensuremath{\accentset{\triangledown}{#1}}}

\newcommand{\gencofid}[1]{\ensuremath{\accentset{\medcircle_{\star}}{#1}}}

\newcommand{\logfid}[1]{\ensuremath{\accentset{\medcircle_{\mathrm{log}}}{#1}}}

\makeatletter
\@ifpackageloaded{tensor}% tensor is a package for a better typesetting of tensors
{

}{%

}
\makeatother

\makeatletter
\@ifpackageloaded{tensor}% tensor is a package for a better typesetting of tensors
{

}{%

}
\makeatother

\newcommand{\absnorm}[1]{\ensuremath{\left|#1\right|}}

\makeatletter
\@ifundefined{volume}{%
}%
{%
}
\makeatother

\newcommand{\tensortensor}[2]{\ensuremath{#1 \otimes #2}}

\makeatletter

\@ifpackageloaded{MnSymbol} % : as binary operator needs MnSymbol package
{
\newcommand{\tensordot}[2]{\ensuremath{#1 \vdotdot #2}} 
}{%
\newcommand{\tensordot}[2]{\ensuremath{#1 : #2}} 
}

\@ifpackageloaded{MnSymbol} % : as binary operator needs MnSymbol package
{
   
}{%
   
}
\makeatother

\newcommand{\vectordot}[2]{\ensuremath{#1 \bullet #2}}

\newcommand{\tensorschur}[2]{\ensuremath{#1 \circ #2}} % Schur/Hadamard product

\newcommand{\tensorf}[1]{{\mathfrak{#1}}}

\newcommand{\Mat}{\R^{d\times d}}
\newcommand{\SymMat}{\Mat_{\rm sym}}
\DeclareMathOperator{\diag}{diag}

\newtheorem{theorem}{Theorem}
\newtheorem{lemma}{Lemma}
\newtheorem{definition}{Definition}

\title[Representation formula for logarithmic derivative]{A new representation formula for the logarithmic corotational derivative---a case study in application of commutator based functional calculus}

\date{\today}

\author{Michal Bathory}
\address{
Faculty of Mathematics and Physics\\
Charles University\\
Sokolovsk\'a 83\\
Praha 8 -- Karl\'{\i}n\\
CZ 186\;75\\
Czech Republic
}
\email{bathory@karlin.mff.cuni.cz}

\author{Miroslav Bul\'{\i}\v{c}ek}
\address{
Faculty of Mathematics and Physics\\
Charles University\\
Sokolovsk\'a 83\\
Praha 8 -- Karl\'{\i}n\\
CZ 186\;75\\
Czech Republic
}
\email{mbul8060@karlin.mff.cuni.cz}

\author{Josef M\'alek}
\address{
Faculty of Mathematics and Physics\\
Charles University\\
Sokolovsk\'a 83\\
Praha 8 -- Karl\'{\i}n\\
CZ 186\;75\\
Czech Republic
}
\email{malek@karlin.mff.cuni.cz}

\author{V\'{\i}t Pr\r{u}\v{s}a}
\address{
Faculty of Mathematics and Physics\\
Charles University\\
Sokolovsk\'a 83\\
Praha 8 -- Karl\'{\i}n\\
CZ 186\;75\\
Czech Republic
}
\email{prusv@karlin.mff.cuni.cz}

\thanks{The authors thank the Czech Science Foundation, grant number 25-16592S, for its support.}

\begin{document}

\begin{abstract}
  The logarithmic corotational derivative is a key concept in rate-type constitutive relations in continuum mechanics. The derivative is defined in terms of the logarithmic spin tensor, which is a skew-symmetric tensor/matrix given by a relatively complex formula. Using a newly developed commutator based functional calculus, we derive a new representation formula for the logarithmic spin tensor. In addition to the result on the logarithmic corotational derivative we also use the newly developed functional calculus to answer some problems regarding the matrix logarithm and the monotonicity of stress-strain relations. These results document that the commutator based functional calculus is of general use in tensor/matrix analysis, and that the calculus allows one to seamlessly work with tensor/matrix valued functions and their derivatives. 
\end{abstract}

\maketitle

%\addtocontents{toc}{\protect\begin{multicols}{2}} % workaround for table of contents in two columns in amsart documentclass, see also the corresponding \addtocontents{toc}{\protect\end{multicols}} line at the end of the document
%\tableofcontents

%\linenumbers

\section{Introduction}
\label{sec:introduction}

The concept of objective time derivative of a tensor quantity is ubiquitous in continuum mechanics especially in theory of elasto-plastic solids, theory of hypo-elastic solids and theory of viscoelastic fluids, see, for example, \cite{truesdell.c.noll.w:non-linear*1} and~\cite{oldroyd.jg:on}. During the development of continuum mechanics the concept of objective derivative of a tensor quantity has been discussed from many viewpoints, see~\cite{szabo.l.balla.m:comparison}, \cite{marsden.je.hughes.tjr:mathematical}, \cite{kolev.b.desmorat.r:objective}, \cite{aubram.d:notes*1} and~\cite{neff.p.holthausen.s.ea:hypo-elasticity} to name a few, and several objective derivatives were introduced. Popular choices of objective derivatives are the \emph{upper-convected derivative}
\begin{equation}
  \label{eq:1}
  \fid{\overline{\generictensor}} =_{\bydefinition} \dd{\generictensor}{t} - \gradvl \generictensor - \generictensor \transpose{\gradvl},
\end{equation}
introduced by~\cite{oldroyd.jg:on}, and a family of \emph{corotational derivatives} defined by the formula
\begin{equation}
  \label{eq:3}
  \gencofid{\overline{\generictensor}} =_{\bydefinition} \dd{\generictensor}{t} + \generictensor \genspin - \genspin \generictensor.
\end{equation}
Here $\dd{}{t} = \pd{}{t} + \vectordot{\vec{v}}{\nabla}$ denotes the standard material time derivative, $\gradvl$ denotes the (Eulerian) velocity gradient, and the symbol~$\genspin$ denotes a skew-symmetric \emph{spin tensor}. A particular choice of spin tensor could be $\genspin =_{\bydefinition} \gradasym$, where $\gradasym$ is the skew-symmetric part of the (Eulerian) velocity gradient. This choice leads to the Jaumann--Zaremba corotational derivative.
% \begin{equation}
%   \label{eq:2}
%   \jfid{\overline{\generictensor}} =_{\bydefinition} \dd{\generictensor}{t} + \generictensor \gradasym - \gradasym \generictensor.
% \end{equation}
Another important \emph{corotational} objective time derivative is the \emph{logarithmic corotational derivative} introduced in the series of works by~\cite{xiao.h.bruhns.ot.ea:logarithmic,xiao.h.bruhns.ot.ea:objective,xiao.h.bruhns.ot.ea:explicit}, see also \cite{bruhns.ot.meyers.a.ea:on}. In this case the spin $\genspin$ in~\eqref{eq:3} is chosen as $\genspin =_{\bydefinition} \logspin$, where
\begin{equation}
  \label{eq:4}
  \logspin
    =_{\bydefinition}
    \gradasym
    +
    \sum_{\substack{\sigma, \tau =1 \\ \sigma \not = \tau}}^3
    \left[
      \left(
        \frac{1 + \frac{b_\sigma}{b_\tau}}{1 - \frac{b_\sigma}{b_\tau}}
        +
        \frac{2}{\ln \frac{b_\sigma}{b_\tau}}
      \right)
      \tensorq{P}_\sigma
      \gradsym
      \tensorq{P}_\tau
    \right]
    ,
  \end{equation}
  and where $b_{\sigma}$ denotes the $\sigma$-th eigenvalue of the left Cauchy--Green tensor $\lcg =_{\bydefinition} \fgrad \transpose{\fgrad}$, and where $\tensorq{P}_\sigma$  denotes the projection to the corresponding eigenspace, see~\cite[Equation 41]{xiao.h.bruhns.ot.ea:logarithmic} for in-depth discussion. The key property of the logarithmic corotational derivative is that \emph{the logarithmic corotational derivative of the Hencky strain tensor $\henckystrain$ yields directly the symmetric part of the (Eulerian) velocity gradient~$\gradsym$}, that is we have the identity
  \begin{equation}
    \label{eq:5}
    \logfid{\overline{\henckystrain}}
    =
    \gradsym,
  \end{equation}
  where the Hencky strain $\henckystrain$ is defined by the formula 
  \begin{equation}
    \label{eq:6}
    \henckystrain
    =_\bydefinition
    \frac{1}{2}
    \ln \lcg
    .
  \end{equation}
  Moreover, the pair Hencky strain--logarithmic corotational derivative is the only strain measure--corotational derivative pair that has the property~\eqref{eq:5}. Thanks to~\eqref{eq:5} the logarithmic derivative is especially suitable for rate-type formulations of elasticity, see, for example, \cite{xiao.h.bruhns.ot.ea:logarithmic,xiao.h.bruhns.ot.ea:natural}, and for other applications in continuum mechanics as well.

  Given the prominent position of the logarithmic derivative in continuum mechanics, the derivative has been the subject of thorough research. In particular, the representation formula~\eqref{eq:4} for logarithmic spin is of interest. In order to find the logarithmic spin, the representation formula~\eqref{eq:4} requires one to find the spectral decomposition of $\lcg$, which might be inconvenient. Using the Cayley--Hamilton theorem, \cite{xiao.h.bruhns.ot.ea:logarithmic} gave an alternative representation to the logarithmic spin, \cite[Equation 45]{xiao.h.bruhns.ot.ea:logarithmic}, but this formula still requires one to work directly with the eigenvalues of $\lcg$, which might be cumbersome in symbolic manipulations. In what follows, we derive a new representation formula for the logarithmic spin tensor $\logspin$, while the new representation formula is based on the newly developed \emph{commutator based functional calculus}.

  In addition to the result on the logarithmic corotational derivative, we also use the newly developed \emph{commutator based functional calculus} to answer several problems regarding the matrix logarithm and the monotonicity of stress-strain relations. These results document that the \emph{commutator based functional calculus} is of general use in tensor/matrix analysis, and that the calculus allows one to seamlessly work with tensor/matrix valued functions and their derivatives. 

  \section{Representation formula for the logarithmic spin}
  \label{sec:repr-form-logar}
  The new representation formula for the logarithmic spin~$\logspin$ exploits the commutator operator. The use of commutator operator is by no means surprising, since the commutator operator provides a characterization of the most important qualitative feature behind the tensor/matrix calculus---the fact that two matrices in general do not commute. For a given matrix $\generictensor$ the commutator operator $\ad_{\generictensor}$ on the space of matrices is defined as
  \begin{equation}
    \label{eq:8}
    \ad_{\generictensor} [\tensorq{X}] =_{\bydefinition} \generictensor \tensorq{X} - \tensorq{X} \generictensor.
  \end{equation}
  Utilising the commutator operator we can prove the following.
  
  \begin{theorem}[Representation formula for the logarithmic spin]
    \label{thr:1}
    Let $\gencofid{\overline{\left( \cdot \right)}}$ denote the corotational derivative~\eqref{eq:3} for which the derivative of the Hencky strain tensor $\henckystrain$, $\henckystrain =_{\bydefinition} \frac{1}{2} \ln \lcg$, yields the symmetric part of the velocity gradient $\gradsym$, that is it holds
    \begin{equation}
      \label{eq:26}
      \gencofid{\overline{\henckystrain}}
      =
      \gradsym.
    \end{equation}
    Then the spin~$\genspin$ in~\eqref{eq:3} is the logarithmic spin $\logspin$, and it is given by the formula
    \begin{equation}
      \label{eq:7}
      \logspin
      =
      \gradasym
      -
      \sigma\left( \ad_{\henckystrain} \right)
      \gradsym
      ,
    \end{equation}
    where $\ad_{\henckystrain}$ is the \emph{commutator operator} associated to $\henckystrain$, and where $\sigma$ denotes the function
    \begin{equation}
      \label{eq:9}
      \sigma(x) =_{\bydefinition} \coth x-\frac{1}{x},
    \end{equation}
    where $\coth x$ is the hyperbolic cotangent function $\coth x =_{\bydefinition} \frac{\exponential{x} + \exponential{-x}}{\exponential{x} - \exponential{-x}}$. The symbol $\sigma\left( \ad_{\henckystrain} \right)$ is defined by the corresponding \emph{formal power series} for the hyperbolic cotangent function, see~\cite[Formula 4.5.67]{abramowitz.m.stegun.ia:handbook}, that is we set
  \begin{equation}
    \label{eq:10}
    \sigma\left( \ad_{\henckystrain} \right)
    =
    \left.
      \left(
        \sum_{n=0}^{+\infty}
        \frac{2^{2n} \mathrm{B}_{2n} x^{2n-1}}{(2n)!}
        -
        \frac{1}{x}
      \right)
    \right|_{x = \ad_{\henckystrain} }
    =
    \sum_{n=1}^{+\infty}
    \frac{2^{2n} \mathrm{B}_{2n} \left( \ad_{\henckystrain} \right)^{2n-1}}{(2n)!},
  \end{equation}
  where $\mathrm{B}_k$ is the $k$-th Bernoulli number. 
\end{theorem}

As shown by \cite[Section 2]{xiao.h.bruhns.ot.ea:logarithmic}, there exists only one corotational derivative with the property~\eqref{eq:26}, therefore the formula~\eqref{eq:7} for the logarithmic spin must yield the same logarithmic spin tensor as that identified by~\cite{xiao.h.bruhns.ot.ea:logarithmic}, who have, however, used the representation~\eqref{eq:4}. Note that the power series in~\eqref{eq:10} is so far only a formal one, the power series can not be used beyond its radius of convergence. Interestingly, the function $\sigma$ introduced in~\eqref{eq:9} is \emph{defined for all} $x \in \R$. (Provided that it is defined by its limit at $x=0$.) This indicates that formulae based on the formal power series might be given another interpretation applicable beyond the radius of convergence of the corresponding power series.

Such an alternative interpretation is possible if we restrict ourselves to \emph{symmetric} tensors/matrices. The alternative interpretation is based on the newly developed \emph{commutator based functional calculus}, which is briefly introduced in Section~\ref{sec:conclusion}, and which is studied in detail in the rigorous follow-up work~\cite{bathory-calculus}. However, in the present work we, for the sake of clarity and simplicity, provide a straightforward proof of~\eqref{eq:7} based on manipulations with formal power series, and we shall tacitly assume that all the formal manipulations can be later rigorously justified in~\emph{commutator based functional calculus} and extended beyond the radius of convergence of the corresponding formal power series.

\subsection{Summary of known formulae for the commutator operator}
  \label{sec:form-comm-op}
  Before we proceed with the proof of Theorem~\ref{thr:1} we recall several formulae for the commutator. Clearly, the formulae listed below indicate the central role of the commutator operator in tensor/matrix analysis. 
  \begin{lemma}[Formulae for the commutator operator]
    \label{lm:1}
    Let $\generictensor, \tensorq{X} \in \Mat$ be arbitrary given matrices, and let $\tensorf{f}$ be a given differentiable tensor/matrix valued function, and let $s \in \Z$. Then
    \begin{subequations}
      \label{eq:29}
      \begin{align}
        \label{eq:28}
        \pd{\exponential{\generictensor}}{\generictensor} \left[ \tensorq{X} \right]
        &=
        \exponential{\generictensor} \frac{1 - \exponential{-\ad_{\generictensor}}}{\ad_{\generictensor}} [\tensorq{X}],
        \\
        \label{eq:23}
        \pd{\ln \generictensor}{\generictensor}
        \left[
        \tensorq{X}
        \right]
        &=
        \left.
        \left(
        \frac{\ad_{\tensorq{Y}}}{1 - \exponential{-\ad_{\tensorq{Y}}}} 
        \left[
        \exponential{- \tensorq{Y}}
        \tensorq{X}
        \right]
        \right)
        \right|_{\tensorq{Y} = \ln \generictensor}
        ,
        \\
        \label{eq:20}
        \exponential{s\tensorq{X}} \tensorq{Y} \exponential{-s\tensorq{X}} &= \exponential{s\ad_{\tensorq{X}}} \left[ \tensorq{Y} \right],
        \\
        \label{eq:18}
        \ad_{\tensorf{f}(\generictensor)} [\tensorq{X}] &= \pd{\tensorf{f}(\generictensor)}{\generictensor} \left[ \ad_{\generictensor} \left[ \tensorq{X} \right]\right],
      \end{align}
    \end{subequations}
    provided that the expressions on either side of the identities are well-defined, in particular provided that the formal power series converge.
  \end{lemma}
  In Lemma~\ref{lm:1} we can indeed work with arbitrary matrices (not necessarily symmetric) as long as we work with the functional calculus based on power series expansions, that is if we define
  \begin{equation}
    \label{eq:15}
    \exponential{\generictensor} =_{\bydefinition} \sum_{n=0}^{+\infty} \frac{\generictensor^n}{n!},
  \end{equation}
  and 
  \begin{equation}
    \label{eq:16}
    \ln \generictensor =_{\bydefinition} \sum_{n=1}^{+\infty} (-1)^{n+1} \frac{\left( \generictensor - \identity \right)^n}{n},
  \end{equation}
  whenever the power series converges. (The same approach is taken for the tensor function $\tensorf{f}$---it is defined by its power series whenever the series converges.) Furthermore, in Lemma~\ref{lm:1} we work with the standard G\^ateaux derivative of tensor/matrix valued function~$\tensorf{f}$ at point $\generictensor$ in the direction $\tensorq{X}$, that is we denote
  \begin{equation}
    \label{eq:17}
    \pd{\tensorf{f} \left( \generictensor \right)}{\generictensor}
    [\tensorq{X}]
    =
    _{\bydefinition}
    \left.
      \dd{}{s}
      \tensorf{f} \left(\generictensor + s \tensorq{X} \right)
    \right|_{s=0}
    .
  \end{equation}

  The formula~\eqref{eq:28} for the G\^ateaux derivative of tensor/matrix exponential is a well-known formula especially in Lie group theory, see, for example, \cite[Theorem 5.4]{hall.b:lie}. Here the symbol $\frac{1 - \exponential{-\ad_{\generictensor}}}{\ad_{\generictensor}}$ is again defined by a formal power series, that is we first introduce the function
  \begin{equation}
    \label{eq:19}
    \eta(x) =_{\bydefinition} \frac{\exponential{x}-1}{x},
  \end{equation}
  and using the known power series for the real function $\eta$ we then define
  \begin{equation}
    \label{eq:25}
    \eta(-\ad_{\generictensor}) =_{\bydefinition}
    \left.
      \left(
        \frac{1 - \sum_{n=0}^{+\infty} \frac{\left(-x\right)^n}{n!}}{x}
      \right)
    \right|_{x = \ad_{\generictensor}}
    =
    \sum_{n=0}^{+\infty} \frac{\left(-1\right)^n}{\left(n+1\right)!} \left( \ad_{\generictensor} \right)^n.
  \end{equation}

  Concerning the inverse function to the exponential---the matrix logarithm---the chain rule gives
  \begin{equation}
    \label{eq:22}
    \foidentity
    \left[
      \tensorq{X}
    \right]
    =
    \pd{\generictensor}{\generictensor}
    \left[
      \tensorq{X}
    \right]
    =
    \pd{\exponential{\ln \generictensor}}{\generictensor}
    \left[
      \tensorq{X}
    \right]
    =
    \left.
      \pd{\exponential{\tensorq{Y}}}{\tensorq{Y}}
    \right|_{\tensorq{Y} = \ln \generictensor}
    \left[
      \pd{\ln \generictensor}{\generictensor}
      \left[
        \tensorq{X}
      \right]
    \right]
    =
    \left.
      \exponential{\tensorq{Y}} \frac{1 - \exponential{-\ad_{\tensorq{Y}}}}{\ad_{\tensorq{Y}}}
    \right|_{\tensorq{Y} = \ln \generictensor}
    \left[
      \pd{\ln \generictensor}{\generictensor}
      \left[
        \tensorq{X}
      \right]
    \right]
    ,
  \end{equation}
  where $\foidentity$ denotes the fourth order identity tensor, and where we have used the formula~\eqref{eq:28} for the derivative of exponential function. Equation~\eqref{eq:22} can be formally solved by a simple manipulation, and we see that~\eqref{eq:22} implies~\eqref{eq:23}, where the symbol
  $
  \frac{\ad_{\tensorq{Y}}}{1 - \exponential{- \ad_{\tensorq{Y}}}} 
  $
  is again interpreted via the corresponding power series. Consequently, we have an explicit expression for the derivative of logarithm function as well.

  Formula~\eqref{eq:20} is the \emph{Campbell formula}, see, for example, \cite[Proposition~3.35]{hall.b:lie}, where the symbol $\exponential{\ad_{\tensorq{X}}}$ is again defined via the formal power series. (Note that $s \ad_{\tensorq{X}} = \ad_{s \tensorq{X}}$.) Finally, if~$f$ is a function with a given power series expansion $f(x) = \sum_{n=0}^{+\infty} f_nx^n$, where the coefficients $\left\{ f_n \right\}_{n=0}^{+\infty}$ are constants, then it is straightforward to see that for the corresponding tensor/matrix function $\tensorf{f}(\generictensor) = \sum_{n=0}^{+\infty} f_n\generictensor^n$ we have  
  \begin{equation}
    \label{eq:21}
    \ad_{\tensorf{f}(\generictensor)} [\tensorq{X}]
    =
    \left(
      \sum_{n=0}^{+\infty} f_n \generictensor^n
    \right)
    \tensorq{X}
    -
    \tensorq{X}
    \left(
      \sum_{n=0}^{+\infty} f_n \generictensor^n
    \right)
    .
  \end{equation}
  On the other hand using the definition of G\^ateaux derivative we have
  \begin{equation}
    \label{eq:27}
    \pd{\tensorf{f}(\generictensor)}{\generictensor} \left[ \ad_{\generictensor} \left[ \tensorq{X} \right]\right]
    =
    \left.
      \dd{}{s}
      \tensorf{f} \left(\generictensor + s  \ad_{\generictensor} \left[ \tensorq{X} \right] \right)
    \right|_{s=0}
    =
    \left.
      \dd{}{s}
      \left(
        f_0 \identity
        +
        f_1
        \left(
          \generictensor
          +
          s
          \ad_{\generictensor} \left[ \tensorq{X} \right]
        \right)
        +
        f_2
        \left(
          \generictensor
          +
          s
          \ad_{\generictensor} \left[ \tensorq{X} \right]
        \right)^2
        +
        \cdots
      \right)
    \right|_{s=0}
    .
  \end{equation}
  We calculate the derivatives of the powers,
  \begin{multline}
    \label{eq:99}
    \left.
      \dd{}{s}
      \left(
        \generictensor + s \ad_{\generictensor} \left[ \tensorq{X} \right]
      \right)^n
    \right|_{s=0}
    =
    \left.
      \dd{}{s}
      \left(
        \generictensor^n
        +
        s
        \sum_{k=0}^{n-1}
        \generictensor^k
        \ad_{\generictensor} \left[ \tensorq{X} \right]
        \generictensor^{n-1-k}
        +
        \cdots
      \right)
    \right|_{s=0}
    =
    \sum_{k=0}^{n-1}
    \generictensor^{k+1}
    \tensorq{X}
    \generictensor^{n-1-k}
    -
    \sum_{k=0}^{n-1}
    \generictensor^k
    \tensorq{X}
    \generictensor^{n-k}
    \\
    =
    \sum_{k=1}^{n}
    \generictensor^{k}
    \tensorq{X}
    \generictensor^{n-k}
    -
    \sum_{k=0}^{n-1}
    \generictensor^k
    \tensorq{X}
    \generictensor^{n-k}
    =
    \generictensor^n
    \tensorq{X}
    -
    \tensorq{X}
    \generictensor^n
    ,
  \end{multline}
  and using~\eqref{eq:99} in the formula for the G\^ateaux derivative~\eqref{eq:27}, we get the commutator formula~\eqref{eq:21}.

  We note that the presented proofs, among others, exploit the crucial fact that the standard calculus rules such as the derivative of product/sum can be applied as usual, and that the composition of linear operators can be interpreted as multiplication of the corresponding matrices.  

  \subsection{Proof of representation formula in Theorem~\ref{thr:1}}
  \label{sec:proof-repr-form}
  Using Lemma~\ref{lm:1} we can easily derive several important formulae for the matrix logarithm. 
  \begin{lemma}[Identities for matrix logarithm]
    \label{lm:2}
    Let $\generictensor, \tensorq{Y} \in \Mat$ be arbitrary given matrices, and let $p, s \in \Z$ be given integers such that $p - s =1$. Then
    \begin{subequations}
      \label{eq:33}
      \begin{align}
        \label{eq:32}
        \generictensor^p  \tensorq{Y} \generictensor^{-s}
        &=
        \generictensor
        \exponential{s\ad_{\ln \generictensor}} \left[ \tensorq{Y} \right]
          ,
        \\
        \label{eq:37}
        \pd{\ln \generictensor}{\generictensor}
        \left[
        \generictensor^p \tensorq{Y} \generictensor^{-s}
        \right]
        &=
          \frac{
          1
          }
          {
          \frac
          {1-\exponential{-\ad_{\ln \generictensor}}
          }
          {
          \ad_{\ln \generictensor}}
          }
          \exponential{s \ad_{\ln \generictensor}}
          \left[
          \tensorq{Y}
          \right]
        \\
        \label{eq:34}
        \pd{\ln \generictensor}{\generictensor} \left[ \generictensor \tensorq{Y} - \tensorq{Y} \generictensor \right]
        &=
          \ad_{\ln \generictensor} [\tensorq{Y}]
          ,
        \\
        \label{eq:36}
        \pd{\ln \generictensor}{\generictensor} \left[ \generictensor \tensorq{Y} + \tensorq{Y} \generictensor \right]
        &=
          \coth\left(\frac{1}{2} \ad_{ \ln \generictensor} \right)
          \ad_{\ln \generictensor}
          \left[
          \tensorq{Y}
          \right]
        ,
      \end{align}
      where the symbols
      $
      \frac{
        1
      }
      {
        \frac
        {1-\exponential{-\ad_{\ln \generictensor}}
        }
        {
          \ad_{\ln \generictensor}
        }
      }
      \exponential{s \ad_{\ln \generictensor}}$
      and
      $
      \coth\left( \frac{1}{2} \ad_{\ln \generictensor} \right) \ad_{\ln \generictensor}
      $
      are defined via formal power series corresponding to the functions $\frac{1}{\frac{1 - \exponential{-x}}{x}} \exponential{s x}$ and $ \coth\left( \frac{1}{2} x \right) x$ respectively.  The identities hold provided that the expressions on either side of the identities are well-defined, in particular provided that the formal power series converge.

    \end{subequations}
  \end{lemma}
  \begin{proof}
    Using Campbell formula~\eqref{eq:20}, we see that
    \begin{equation}
      \label{eq:30}
      \generictensor^s  \tensorq{Y} \generictensor^{-s}
      =
      \exponential{s\ln \generictensor} \tensorq{Y} \exponential{-s \ln \generictensor}
      =
      \exponential{s\ad_{\ln \generictensor}} \left[ \tensorq{Y} \right],
    \end{equation}
    which immediately implies~\eqref{eq:32}. The second identity~\eqref{eq:37} follows from a simple substitution into the formula for the derivative of matrix logarithm, see~\eqref{eq:23}, and the just derived identity~\eqref{eq:32},
    \begin{equation}
      \label{eq:31}
      \pd{\ln \generictensor}{\generictensor}
      \left[
        \generictensor^p \tensorq{Y} \generictensor^{-s}
      \right]
      =
      \frac{1}
      {
        \frac{
          1 - \exponential{-\ad_{\ln \generictensor}}
        }
        {
          \ad_{\ln \generictensor}
        }
      }
      \left[
        \exponential{- \ln \generictensor}
        \generictensor
        \exponential{s \ad_{\ln \generictensor}}
        \left[
          \tensorq{Y}
        \right]
      \right]
      .
    \end{equation}
    The identity~\eqref{eq:34} is just a special case of~\eqref{eq:18}. Finally, the last identity~\eqref{eq:36} is obtained from~\eqref{eq:37} with $p=1$, $s=0$ and $p=0$, $s=-1$ respectively via a straightforward manipulation
    \begin{equation}
      \label{eq:35}
      \pd{\ln \generictensor}{\generictensor} \left[ \generictensor \tensorq{Y} + \tensorq{Y} \generictensor \right]
      =
      \pd{\ln \generictensor}{\generictensor} \left[ \generictensor \tensorq{Y}  \right]
      +
      \pd{\ln \generictensor}{\generictensor} \left[ \tensorq{Y} \generictensor \right]
      =
      \frac{
        1
      }
      {
        \frac
        {1-\exponential{-\ad_{\ln \generictensor}}
        }
        {
          \ad_{\ln \generictensor}
        }
      }
      \left[
        \tensorq{Y}
      \right]
      +
      \frac{
        1
      }
      {
        \frac
        {1-\exponential{-\ad_{\ln \generictensor}}
        }
        {
          \ad_{\ln \generictensor}
        }
      }
      \exponential{- \ad_{\ln \generictensor}}
      \left[
        \tensorq{Y}
      \right]
      \\
      =
      \frac{
        1
        +
        \exponential{-\ad_{\ln \generictensor}}
      }
      {
        \frac
        {1-\exponential{-\ad_{\ln \generictensor}}
        }
        {
          \ad_{\ln \generictensor}}
      }
      \left[
        \tensorq{Y}
      \right]
      =
      \coth\left( \frac{1}{2} \ad_{ \ln \generictensor} \right)
      \ad_{\ln \generictensor}
      \left[
        \tensorq{Y}
      \right]
      ,
    \end{equation}
    where we have also used trivial identity
    \begin{equation}
      \label{eq:49}
      \tensorf{f} \left( \ad_{\generictensor} \right) \ad_{\generictensor} = \ad_{\generictensor} \tensorf{f} \left( \ad_{\generictensor} \right).
    \end{equation}
  \end{proof}

  \begin{proof}[Proof of Theorem~\ref{thr:1}]
    Having obtained auxiliary results in Lemma~\ref{lm:2} we can return to the original problem. The problem is to find a spin tensor~$\genspin$ and a function $\tensorf{f}$ of the left Cauchy--Green tensor $\lcg$ such that the corotational derivative~\eqref{eq:3} of the function $\tensorf{f}(\lcg)$ yields the symmetric part of the velocity gradient $\gradsym$. In other words we want to find  $\tensorf{f}(\lcg)$ and $\genspin$ such that
  \begin{equation}
    \label{eq:24}
     \gencofid{\overline{\tensorf{f}(\lcg)}} = \gradsym.
   \end{equation}
   Further restrictions on the function $\tensorf{f}$ come from the requirement of $\tensorf{f}$ being a \emph{genuine strain measure}. This means that the function $\tensorf{f}$ should be an isotropic tensorial function, for definition see, for example, \cite{truesdell.c.noll.w:non-linear*1}. Furthermore, the function $\tensorf{f}$ should vanish in the case of no deformation, that is for $\lcg = \identity$ we want
   \begin{subequations}
     \label{eq:38}
     \begin{equation}
       \label{eq:41}
            \left. \tensorf{f} \left(\lcg \right) \right|_{\lcg = \identity} = \tensorzero.
     \end{equation}
   Finally, we want $\tensorf{f}(\lcg)$ to be identical to the small strain tensor $\linstrain$ in the case of small deformations. For small deformations (small displacement gradients) we have $\lcg \approx \identity + 2 \linstrain$, and we want $\left. \tensorf{f} \left(\lcg \right) \right|_{\lcg = \identity + 2 \linstrain} \approx \linstrain$, which in virtue of~\eqref{eq:41} yields a condition on the first derivative of $\tensorf{f}$,
   \begin{equation}
     \label{eq:56}
     \left. \pd{\tensorf{f}}{\lcg} \right|_{\lcg = \identity} \left[ \tensorq{X} \right] = \frac{1}{2} \tensorq{X}.
   \end{equation}
   \end{subequations}   

   We start with the required identity~\eqref{eq:24}, and we try to identify the suitable spin tensor. Differentiating the left-hand side of~\eqref{eq:24} according to the definition of a corotational rate~\eqref{eq:3}, and using the chain rule yields
   \begin{equation}
     \label{eq:39}
     \pd{\tensorf{f}(\lcg)}{\lcg}
     \left[
       \dd{\lcg}{t}
     \right]
     +
     \tensorf{f} \left( \lcg \right)
     \genspin
     -
     \genspin
     \tensorf{f} \left( \lcg \right)
     =
     \gradsym
     .
   \end{equation}
   The upper convected derivative of $\lcg$ is known to vanish, $\fid{\overline{\lcg}} = \tensorzero$, which yields
   \begin{equation}
     \label{eq:40}
     \dd{\lcg}{t} - \gradvl \lcg - \lcg \transpose{\gradvl} = \tensorzero.
   \end{equation}
   (This fact is a simple consequence of the definition of Cauchy--Green tensor $\lcg =_{\bydefinition} \fgrad \transpose{\fgrad}$ and the identity $\dd{\fgrad}{t} = \gradvl \fgrad$.) In~\eqref{eq:40} we use the decomposition of the velocity gradient to the symmetric and the skew-symmetric part, $\gradvl = \gradsym + \gradasym$, and we substitute the so obtained formula for the time derivative $\dd{\lcg}{t}$ into~\eqref{eq:39}, which yields
   \begin{equation}
     \label{eq:42}
     \pd{\tensorf{f}(\lcg)}{\lcg}
     \left[
       \gradsym \lcg
       +
       \lcg \gradsym
     \right]
     +
     \pd{\tensorf{f}(\lcg)}{\lcg}
     \left[
       \gradasym \lcg
       -
       \lcg \gradasym
     \right]
     +
     \tensorf{f} \left( \lcg \right)
     \genspin
     -
     \genspin
     \tensorf{f} \left( \lcg \right)
     =
     \gradsym
     .
   \end{equation}
   We rearrange the terms, and we get
   \begin{equation}
     \label{eq:43}
     \pd{\tensorf{f}(\lcg)}{\lcg}
     \left[
       \gradsym \lcg
       +
       \lcg \gradsym
     \right]
     -
     \gradsym
     =
     -
     \ad_{\tensorf{f}(\lcg)} \left[ \genspin \right]
     +
     \pd{\tensorf{f}(\lcg)}{\lcg}
     \left[
       \lcg \gradasym
       -
       \gradasym \lcg
     \right]
     .
   \end{equation}
    If we denote $\genspin=\gradasym-\tensorq{\Omega}_0$ and search for $\tensorq{\Omega}_0$ instead of $\genspin$, then by virtue of \eqref{eq:18}, the above identity simplifies to
    \begin{equation}
     \label{eq:44}
     \pd{\tensorf{f}(\lcg)}{\lcg}
     \left[
       \gradsym \lcg
       +
       \lcg \gradsym
     \right]
     -
     \gradsym
     =
     \ad_{\tensorf{f}(\lcg)} \left[ \tensorq{\Omega}_0 \right]
     ,
   \end{equation}
   which must hold for arbitrary motion, that is for arbitrary $\gradsym$ and $\lcg$. If we use~\eqref{eq:44} in the special case $\gradsym \equiv \inverse{\lcg}$, and if take the trace on the both sides of~\eqref{eq:44}, then we get
   \begin{equation}
     \label{eq:38}
        2\Tr
        \left(
            \pd{\tensorf{f}(\lcg)}{\lcg}
                \left[
                    \identity
                \right]
        \right)
     =
     \Tr \inverse{\lcg}.
    \end{equation}
    If we restrict ourselves to isotropic functions $\tensorf{f}$, we see that~\eqref{eq:38} can only be satisfied for all $\lcg$ provided that
    \begin{equation}
      \label{eq:46}
        \tensorf{f}(\lcg)=\frac{1}{2}\ln \lcg + f_0 \identity
    \end{equation}
    for some constant $f_0$. However, should $\tensorf{f}$ define a \emph{genuine strain measure}, it is necessary that $\tensorf{f}(\identity)=\tensorzero$, see~\eqref{eq:41}, hence $f_0=0$. The second condition for $\tensorf{f}$ being a genuine strain measure, see~\eqref{eq:56}, is also satisfied by the choice~\eqref{eq:46}. 

    Note that we could have also guessed that $\tensorf{f}(\lcg)=\frac{1}{2}\ln\lcg$ by a simple observation. The equality~\eqref{eq:43}  should hold for arbitrary motion, that is for arbitrary $\gradsym$ and $\gradasym$, which indicates that on the left-hand side we would need
   \begin{equation}
     \label{eq:76}
     \pd{\tensorf{f}(\lcg)}{\lcg}
     \left[ \tensorq{X} \right]
     \sim
     \frac{1}{2} \inverse{\lcg} \tensorq{X},
   \end{equation}
   which indicates that $\tensorf{f}(\lcg) \sim \frac{1}{2}\ln\lcg$. (Here we naively use the standard formula for the real function $\ln x$, $\dd{}{x} \ln x = \frac{1}{x}$.) Indeed, if~\eqref{eq:44} was true and $\gradsym$ would be commuting with $\lcg$, then the left-hand side would vanish, and we would have to deal with the right-hand side only, which could be made zero by a suitable choice of $\genspin$. Unfortunately, the tensor/matrix derivative is more involved to work with, and we also have problems with the commutativity of $\gradsym$ and $\lcg$. Nevertheless, as shown by the argument following~\eqref{eq:44}, the observation that $\tensorf{f}(\lcg)=\frac{1}{2}\ln\lcg$ is correct.

    Returning with the information that $\tensorf{f}(\lcg)=\frac{1}{2}\ln\lcg$ to \eqref{eq:44}, we can apply Lemma~\ref{lm:2}---identity~\eqref{eq:34} and \eqref{eq:36}---and we can rewrite~\eqref{eq:44} as
   \begin{equation}
     \label{eq:45}
     \coth\left(\ad_{\frac{1}{2} \ln \lcg} \right)
     \ad_{\frac{1}{2} \ln \lcg}
     \left[
       \gradsym
     \right]
     -
     \gradsym
     =
     \ad_{ \frac{1}{2} \ln \lcg} \left[ \tensorq{\Omega}_0 \right]
     ,
   \end{equation}
   where we have also used the trivial observation  $\frac{1}{2} \ad_{\generictensor} = \ad_{\frac{1}{2} \generictensor}$. Now we use the commutator based calculus, and we rewrite the last equality as
   \begin{equation}
     \label{eq:75}
     \ad_{ \frac{1}{2} \ln \lcg}
     \underbrace{
     \left(
       \coth\left(\ad_{\frac{1}{2} \ln \lcg} \right)
       -
       \frac{1}{\ad_{\frac{1}{2} \ln \lcg}}
     \right)
     }_{\sigma \left(\ad_{\frac{1}{2} \ln \lcg}\right)}
     \left[
       \gradsym
     \right]
     =
     \ad_{ \frac{1}{2} \ln \lcg} \left[ \tensorq{\Omega}_0 \right]
     .
   \end{equation}
   We note that the round bracket, or the symbol $\sigma (\ad_{\frac{1}{2} \ln \lcg})$, is well-defined by the corresponding power series. (No negative powers are present in the power series.) The quantity~$\gradsym$ in the last equation~\eqref{eq:75} can take arbitrary values, hence the equation holds if and only if we set $\tensorq{\Omega}_0=\sigma (\ad_{\frac{1}{2} \ln \lcg})$, that is
   \begin{equation}
     \label{eq:47}
     \genspin
     =
     \gradasym
     -
     \sigma \left(\ad_{\frac{1}{2} \ln \lcg}\right)
     \left[
       \gradsym
     \right]
     .
   \end{equation}
   This is the representation formula~\eqref{eq:7} in Theorem~\ref{thr:1}.
 \end{proof}

 \section{Applications}
 \label{sec:applications}
 The proposed calculus can be exploited beyond the proof of Theorem~\ref{thr:1}. In what follows we discuss some applications related to some matrix/tensor analysis problems in continuum mechanics. 

 \subsection{Identities involving the derivative of matrix logarithm }
 \label{sec:form-deriv-matr}
 For later use we first derive some additional identities for the matrix logarithm.
 \begin{lemma}[Formulae for the derivative of matrix logarithm]
   \label{lm:4}
   Let $\generictensor \in \Mat$ and~$\tensorq{X} \in \Mat$ be given arbitrary matrices, and let $p, s \in \Z$ be given integers such that $p-s=1$. Then 
   \begin{subequations}
     \label{eq:57}
     \begin{align}
       \label{eq:58}
       \pd{\ln \generictensor}{\generictensor}
       \left[
       \generictensor^p \tensorq{X} \generictensor^{-s}
       -
       \generictensor^{-s} \tensorq{X} \generictensor^p
       \right]
       &=
         \frac{\sinh \left( \frac{p+s}{2} \ad_{\ln \generictensor} \right)}{\sinh \left( \frac{1}{2} \ad_{\ln \generictensor} \right)}
         \ad_{\ln \generictensor} [\tensorq{X}]
          ,
        \\
       \label{eq:59}
       \pd{\ln \generictensor}{\generictensor}
       \left[
       \generictensor^p \tensorq{X} \generictensor^{-s}
       +
       \generictensor^{-s} \tensorq{X} \generictensor^p
       \right]
       &=
         \frac{\cosh \left( \frac{p+s}{2} \ad_{\ln \generictensor} \right)}{\sinh \left( \frac{1}{2} \ad_{\ln \generictensor} \right)}
         \ad_{\ln \generictensor} [\tensorq{X}]
         ,
     \end{align}
   \end{subequations}
   provided that the expressions on either side of the identities are well-defined, in particular provided that the formal power series converge.
 \end{lemma}
 \begin{proof}
   The identities follow from Lemma~\ref{lm:2}, formula~\eqref{eq:37}. Since $p-s = 1$, we see that $ \frac{p + s}{2} = s + \frac{1}{2} = p - \frac{1}{2}$, hence in virtue of~\eqref{eq:37} we have
   \begin{equation}
     \label{eq:60}
     \pd{\ln \generictensor}{\generictensor}
     \left[
       \generictensor^p \tensorq{X} \generictensor^{-s}
       -
       \generictensor^{-s} \tensorq{X} \generictensor^p
     \right]
     =
     \frac{
       1
     }
     {
       \frac
       {1-\exponential{-\ad_{\ln \generictensor}}
       }
       {
         \ad_{\ln \generictensor}}
     }
     \exponential{s \ad_{\ln \generictensor}}
     \left[
       \tensorq{X}
     \right]
     -
     \frac{
       1
     }
     {
       \frac
       {1-\exponential{-\ad_{\ln \generictensor}}
       }
       {
         \ad_{\ln \generictensor}}
     }
     \exponential{-p \ad_{\ln \generictensor}}
     \left[
       \tensorq{X}
     \right]
     =
     \frac
     {
     \exponential{ \left(s + \frac{1}{2}\right) \ad_{\ln \generictensor}}
     -
     \exponential{ -\left(p - \frac{1}{2}\right) \ad_{\ln \generictensor}}
     }
     {
       \exponential{\frac{\ad_{\ln \generictensor}}{2}}-\exponential{-\frac{\ad_{\ln \generictensor}}{2}}
     }
       \ad_{\ln \generictensor}
       \left[
         \tensorq{X}
       \right]
       ,
     \end{equation}
     which yields~\eqref{eq:58}. The second identity~\eqref{eq:59} follows by a similar manipulation.
 \end{proof}
 Now we are ready to investigate some open problems regarding identities involving the derivative of matrix logarithm. In particular, the relation between the expression
 \begin{equation*}
   \pd{\ln\generictensor}{\generictensor}\left[\generictensor \tensorq{X} + \tensorq{X} \generictensor \right]
 \end{equation*}
 and $2 \tensorq{X}$ is of interest in many works on continuum mechanics, see, for example, \cite{holthausen.s.ea:constitutive} and \cite{neff.p.holthausen.s.ea:hypo-elasticity}. A natural question is whether these two expressions are equal, that is what are the conditions for the validity of the equality
 \begin{equation}
  \label{eq:13}
  \pd{\ln\generictensor}{\generictensor}\left[\generictensor \tensorq{X} + \tensorq{X} \generictensor \right] = 2 \tensorq{X}.
\end{equation}
It turns out that this question is straightforward to answer using the just developed calculus. Indeed, with the help of~\eqref{eq:36} and~\eqref{eq:59} we can prove the following lemma.

\begin{lemma}[Formula for the derivative of matrix logarithm]
  \label{lm:3}
  Let $\generictensor \in \SymMat$ be a symmetric positive definite matrix and let $\tensorq{X} \in \Mat$ be a general matrix. Then
  \begin{equation}
    \label{eq:50}
    \pd{\ln\generictensor}{\generictensor}\left[\generictensor \tensorq{X} + \tensorq{X} \generictensor \right]
    =
    2 \tensorq{X}
  \end{equation}
  if and only if $\tensorq{X}$ commutes with $\generictensor$, provided that all expressions are well-defined, in particular provided that the formal power series converge.
  % 
  % Furthermore
  % \begin{equation}
  %   \label{eq:51}
  %   \pd{\ln\generictensor}{\generictensor}\left[\generictensor^{\frac{1}{2}} \tensorq{X} \generictensor^{\frac{1}{2}} \right]
  %   =
  %   \tensorq{X}
  % \end{equation}
  % if and only if $\tensorq{X}$ commutes with $\generictensor$.
\end{lemma}

Formula~\eqref{eq:50} refines \cite[Proposition~A.31]{neff.p.holthausen.s.ea:hypo-elasticity}, and it finishes the calculation in \cite[Equation~A.153]{neff.p.holthausen.s.ea:hypo-elasticity}, hence it nicely demonstrates the effectiveness of the commutator calculus over the traditional approach via spectral decompositions.

\begin{proof}
Concerning the proof of formula~\eqref{eq:50} it is enough to employ equality~\eqref{eq:36} and use the identity
\begin{equation}
  \label{eq:52}
  \left(\coth \frac{x}{2} \right) x = 2 + \gamma(x) x^2,
\end{equation}
where $\gamma$ denotes the function
\begin{equation}
  \label{eq:53}
  \gamma (x)
  =_{\bydefinition}
  \begin{cases}
    \frac{\left(\coth \frac{x}{2} \right) x - 2}{x^2}, & x \not = 0,\\
    \frac{1}{6}, & x = 0.
  \end{cases}
\end{equation}
We note that $\gamma$ is a continuous and \emph{positive} function. Using the just introduced notation, we can rewrite~\eqref{eq:36} as
\begin{equation}
  \label{eq:54}
  \pd{\ln \generictensor}{\generictensor} \left[ \generictensor \tensorq{X} + \tensorq{X} \generictensor \right]
  =
  \left(
     2 + \gamma \left( \ad_{\ln \generictensor} \right) \ad_{\ln \generictensor}^2
  \right)
  \left[
    \tensorq{X}
  \right]
  .
\end{equation}
Thus we see that~\eqref{eq:50} holds if and only if
\begin{equation}
  \label{eq:65}
  \gamma \left( \ad_{\ln \generictensor} \right) \ad_{\ln \generictensor}^2 \left[ \tensorq{X}\right] = \tensorzero. 
\end{equation}
This equation is easy to solve in the proposed formalism. First we note that in virtue or the \emph{positivity} of function~$\gamma(x)$ we can find the inverse to the linear operator
\begin{equation}
  \label{eq:64}
  \tensorq{U} \in \Mat \mapsto \gamma (\ad_{\tensorq{Y}})[\tensorq{U}] \in \Mat
\end{equation}
as $\frac{1}{\gamma} \left( \ad_{\tensorq{Y}} \right)$. Consequently the linear operator $\gamma (\ad_{\tensorq{Y}})$ is a bijection, hence the only matrix $\tensorq{U}$ that is mapped to the zero matrix~$\tensorzero$ is the zero matrix. Thus in order to solve~\eqref{eq:65} we actually need to solve
\begin{equation}
  \label{eq:11}
  \ad_{\ln \generictensor}^2 \left[ \tensorq{X}\right] = \tensorzero.
\end{equation}
A simple algebraic manipulation reveals that
\begin{equation}
  \label{eq:55}
  \tensordot{\left( \ad_{\ln \generictensor}^2 \left[ \tensorq{X}\right] \right)}{\tensorq{X}}
  =
  \tensordot{\left( \ad_{\ln \generictensor} \left[ \tensorq{X}\right] \right)}{ \left( \ad_{\ln \generictensor} \left[ \tensorq{X} \right] \right)}
  ,
\end{equation}
where the dot product on the space of matrices is the standard one, that is $\tensordot{\tensorq{U}}{\tensorq{V}} = _{\bydefinition} \Tr \left( \tensorq{U} \transpose{\tensorq{V}} \right)$, and where the symbol $\absnorm{\cdot}$ denotes the norm induced by this dot product, $\absnorm{\tensorq{U}} =_{\bydefinition} \left(\tensordot{\tensorq{U}}{\tensorq{V}}\right)^{\frac{1}{2}}$. (Note that here we are exploiting the symmetry of $\generictensor$.) The equality~\eqref{eq:55} with the equation~\eqref{eq:11} then imply that
\begin{equation}
  \label{eq:12}
  \absnorm{\ad_{\ln \generictensor} \left[ \tensorq{X}\right]}^2 = 0,
\end{equation}
which means that $\ln \generictensor$ and hence also $\generictensor$ must commute with $\tensorq{X}$.  

\end{proof}

\subsection{Monotonicity of stress--strain relations}
\label{sec:monot-stress-stra}
Yet another problem easy to solve using the proposed calculus is related to the analysis of \emph{corotational stability postulate}, see~\cite{holthausen.s.ea:constitutive}. The core problem is the following. For the given \emph{isotropic tensorial function} $\tensorf{f}$ that maps \emph{symmetric} tensors/matrices to \emph{symmetric} tensors/matrices one wants to show the equivalence between the characterisation
\begin{equation}
  \label{eq:69}
  \forall \generictensor \in {\SymMat}_+, \forall \tensorq{X} \in \SymMat \setminus \left\{ \tensorzero \right\}\colon \qquad
  \tensordot{
    \left(
      \pd{\tensorf{f} \left( \ln \generictensor \right)}{\generictensor}
      \left[
        \generictensor \tensorq{X}
        +
        \tensorq{X} \generictensor
      \right]
    \right)
  }
  {
    \tensorq{X}
  }
  >
  0
  ,
\end{equation}
and the characterisation
\begin{equation}
  \label{eq:70}
  \forall \tensorq{G} \in \SymMat, \forall \tensorq{X} \in \SymMat \setminus \left\{ \tensorzero \right\}\colon \qquad
  \tensordot{
    \left(
      \pd{\tensorf{f}(\tensorq{G})}{\tensorq{G}}
      \left[
        \tensorq{X}
      \right]
    \right)
  }
  {
    \tensorq{X}
  }
  >
  0
  ,
\end{equation}
see~\cite[Equation 1.11]{neff.p.holthausen.s.ea:hypo-elasticity}.

Before we start with the proof equivalence between~\eqref{eq:69} and~\eqref{eq:70}, we recall that if $\tensorf{f}$ is an isotropic tensorial function, then the representation theorem for isotropic tensorial functions, see, for example, \cite[Section 12]{truesdell.c.noll.w:non-linear*1}, implies that
\begin{equation}
  \label{eq:63}
  \tensorf{f}(\generictensor) \generictensor = \generictensor \tensorf{f}(\generictensor).
\end{equation}
This property allows us to show the following lemma.
\begin{lemma}[Commutativity of $\pd{\tensorf{f}(\generictensor)}{\generictensor}$ and $\ad_{\generictensor}$ for isotropic tensorial functions]
  \label{lm:5}
  Let $\tensorf{f} : \SymMat \mapsto \SymMat$ be a differentiable isotropic tensorial function, $\generictensor \in \SymMat$ be a symmetric matrix, and let $\tensorq{Y} \in \Mat$ be an arbitrary matrix. Then
  \begin{equation}
    \label{eq:81}
    \ad_{\generictensor}
    \left[
      \pd{\tensorf{f} \left(\generictensor \right)}{\generictensor} \left[ \tensorq{Y} \right]
    \right]
    =
    \pd{\tensorf{f}(\generictensor)}{\generictensor} \left[ \ad_{\generictensor} \left[ \tensorq{Y} \right]\right]
    .
  \end{equation}
  
\end{lemma}

\begin{proof}
  Differentiating the property~\eqref{eq:63}, we see that
  \begin{equation}
    \label{eq:66}
    \pd{\tensorf{f} \left(\generictensor \right)}{\generictensor} \left[ \tensorq{Y} \right] \generictensor
    +
    \tensorf{f} \left(\generictensor \right)
    \tensorq{Y}
    =
    \tensorq{Y}
    \tensorf{f} \left(\generictensor \right)
    +
    \generictensor
    \pd{\tensorf{f} \left(\generictensor \right)}{\generictensor} \left[ \tensorq{Y} \right], 
  \end{equation}
  which means that
  \begin{equation}
    \label{eq:74}
    \ad_{\tensorf{f} \left( \generictensor \right)}
    \left[
      \tensorq{Y}
    \right]
    =
    \ad_{\generictensor}
    \left[
      \pd{\tensorf{f} \left(\generictensor \right)}{\generictensor} \left[ \tensorq{Y} \right]
    \right]
    .
  \end{equation}
  On the other hand we have the identity~\eqref{eq:18}, that is
  \begin{equation}
    \label{eq:78}
    \ad_{\tensorf{f}(\generictensor)} [\tensorq{Y}]
    =
    \pd{\tensorf{f}(\generictensor)}{\generictensor} \left[ \ad_{\generictensor} \left[ \tensorq{Y} \right]\right].
  \end{equation}
  Combining~\eqref{eq:74} and~\eqref{eq:78} we see that for isotropic tensorial functions we have the identity~\eqref{eq:81}.
\end{proof}

For the symmetric matrices we can also provide a straightforward characterisation of the transpose.

\begin{lemma}[Transpose of $f(\ad_\generictensor)$]
  \label{lm:6}
  Let $\generictensor \in \SymMat$ be a symmetric matrix, $\tensorq{X} \in \Mat$ be an arbitrary matrix, and let $f$ be a function whose formal power series defines the operator $f (\ad_\generictensor)$. Then
  \begin{equation}
    \label{eq:80}
    \transpose{\left( f (\ad_\generictensor) \left[ \tensorq{X} \right] \right)} = f(-\ad_{\generictensor}) \left[ \transpose{\tensorq{X}} \right]. 
  \end{equation}
  Furthermore, let $\tensorq{Y} \in \Mat$ be an arbitrary matrix, then
  \begin{equation}
    \label{eq:85}
    \tensordot{\left( f (\ad_\generictensor) \left[ \tensorq{X} \right] \right)}{\tensorq{Y}}
    =
    \tensordot{\tensorq{X}}{\left( f (\ad_\generictensor) \left[ \tensorq{Y} \right] \right)}
    .
  \end{equation}
\end{lemma}

\begin{proof}
  We see that
  \begin{equation}
    \label{eq:83}
    \transpose{\left( \ad_\generictensor \left[ \tensorq{X}\right] \right)}
    =
    \transpose{\left( \generictensor \tensorq{X} - \tensorq{X} \generictensor \right)}
    =
    \transpose{\tensorq{X}} \generictensor - \generictensor \transpose{\tensorq{X}}
    =
    -
    \ad_\generictensor \left[ \transpose{\tensorq{X}} \right]
    ,
  \end{equation}
  hence for $n \in \N$ we have
  \begin{equation}
    \label{eq:84}
    \transpose{\left( \ad_\generictensor^n \left[ \tensorq{X}\right] \right)}
    =
    \transpose{
      \ad_{\generictensor}
      \left[
        \ad_{\generictensor}
        \left[
          \dots
          \left.
            \ad_{\generictensor}
            \left[
              \tensorq{X}
            \right]
          \right]
          \dots
        \right]
      \right]
    }
    =
    -
    \ad_{\generictensor}
    \left[
      -\ad_{\generictensor}
      \left[
        \dots
        \left.
          -\ad_{\generictensor}
          \left[
            \transpose{\tensorq{X}}
          \right]
        \right]
        \dots
      \right]
    \right]
    =
    \left(
      -
      \ad_\generictensor
    \right)^n
    \left[ \transpose{\tensorq{X}} \right]
    .
  \end{equation}
  Consequently, if~$f$ is a function with a given power series expansion $f(x) = \sum_{n=0}^{+\infty} f_nx^n$, then 
  \begin{multline}
    \label{eq:82}
    \transpose{\left( f (\ad_\generictensor) \left[ \tensorq{X} \right] \right)}
    =
    \transpose{
      \left(
        f_1 \ad_\generictensor \left[ \tensorq{X}\right]
        +
        f_2 \ad_{\generictensor} \left[  \ad_\generictensor \left[ \tensorq{X}\right] \right]
        +
        \cdots
      \right)
    }
    \\
    =
    f_1 \left(-\ad_\generictensor\right) \left[ \transpose{\tensorq{X}} \right]
    +
    f_2 \left(- \ad_{\generictensor} \right) \left[  \left(-\ad_\generictensor\right) \left[ \transpose{\tensorq{X}}\right] \right]
    +
    \cdots
    =
    f (-\ad_\generictensor) \left[ \transpose{\tensorq{X}} \right],
  \end{multline}
  which gives us the first proposition~\eqref{eq:80}. Concerning the second proposition~\eqref{eq:85} we first observe that 
  \begin{equation}
    \label{eq:87}
    \tensordot{\left(\ad_\generictensor \left[ \tensorq{X} \right] \right)}{\tensorq{Y}}
    =
    \tensordot{\tensorq{X}}{\left(\ad_\generictensor \left[ \tensorq{Y}  \right] \right)}
    .
  \end{equation}
  Indeed, the definition of the dot product and the cyclic property of the trace together with the symmetry of $\generictensor$ yield
  \begin{equation}
    \label{eq:98}
    \tensordot{\left(\ad_\generictensor \left[ \tensorq{X} \right] \right)}{\tensorq{Y}}
    =
    \Tr
    \left(
      \left(
        \ad_\generictensor \left[ \tensorq{X} \right]
      \right)
      \transpose{\tensorq{Y}}
    \right)
    =
    \Tr
    \left(
      \left(
        \generictensor \tensorq{X}
        -
        \tensorq{X} \generictensor
      \right)
      \transpose{\tensorq{Y}}
    \right)
    =
   \Tr
    \left(
      \left(
        \tensorq{X} \transpose{\tensorq{Y}} \generictensor
        -
        \tensorq{X} \generictensor \transpose{\tensorq{Y}}
      \right)
    \right)
    =
    \Tr
    \left(
      \tensorq{X}
      \transpose{
        \left(
          \generictensor \tensorq{Y} 
          -
          \tensorq{Y} \generictensor
        \right)
      }
    \right)
    =
    \tensordot{\tensorq{X}}{\left(\ad_\generictensor \left[ \tensorq{Y}  \right] \right)}
    .
  \end{equation}
  By induction we then get
  $
    \tensordot{\left(\ad_\generictensor^n \left[ \tensorq{X} \right] \right)}{\tensorq{Y}}
    =
    \tensordot{\tensorq{X}}{\left(\ad_\generictensor^n \left[ \tensorq{Y} \right] \right)}
$
  for arbitrary $n \in \N$, and the proposition~\eqref{eq:85} is a consequence of the power series representation $f(\ad_{\generictensor}) = \sum_{n=0}^{+\infty} f_n \ad_{\generictensor}^n$.  
\end{proof}

\begin{lemma}[Equivalent characterisation of isotropic tensorial functions functions]
  Let $\tensorf{f} : \SymMat \mapsto \SymMat$ be a differentiable isotropic tensorial function. Let $p, s \in \Z$ be arbitrary integers such that  $p-s=1$. Then
  \begin{equation}
    \label{eq:72}
    \forall \generictensor \in {\SymMat}_+, \forall \tensorq{X} \in \SymMat \setminus \left\{ \tensorzero \right\}\colon \qquad
    \tensordot{
      \left(
        \pd{\tensorf{f} \left( \ln \generictensor \right)}{\generictensor}
        \left[
          \generictensor^p \tensorq{X} \generictensor^{-s}
          +
          \generictensor^{-s} \tensorq{X} \generictensor^p
        \right]
      \right)
    }
    {
      \tensorq{X}
    }
    >
    0
  \end{equation}
  holds if and only if
  \begin{equation}
    \label{eq:73}
    \forall \tensorq{G} \in \SymMat, \forall \tensorq{X} \in \SymMat \setminus \left\{ \tensorzero \right\}\colon \qquad
      \tensordot{
    \left(
      \pd{\tensorf{f}(\tensorq{G})}{\tensorq{G}}
      \left[
        \tensorq{X}
      \right]
    \right)
  }
  {
    \tensorq{X}
  }
  >
  0
  .
  \end{equation}
  In particular, if we set $p=1$ and $s=0$ we have the equivalence between~\eqref{eq:69} and~\eqref{eq:70}.
\end{lemma}

\begin{proof}
  We manipulate~\eqref{eq:72} using the chain rule and the identity~\eqref{eq:59} from Lemma~\ref{lm:4}, and we get
  \begin{equation}
    \label{eq:71}
    \tensordot{
      \left(
        \pd{\tensorf{f} \left( \ln \generictensor \right)}{\generictensor}
        \left[
          \generictensor^p \tensorq{X} \generictensor^{-s}
          +
          \generictensor^{-s} \tensorq{X} \generictensor^p
        \right]
      \right)
    }
    {
      \tensorq{X}
    }
    =
    \tensordot{
      \left(
        \left. \pd{\tensorf{f} \left( \tensorq{G} \right)}{\tensorq{G}} \right|_{\tensorq{G} = \ln \generictensor}
        \left[
          \pd{\ln \generictensor}{\generictensor}
          \left[
            \generictensor^p \tensorq{X} \generictensor^{-s}
            +
            \generictensor^{-s} \tensorq{X} \generictensor^p
          \right]
        \right]
      \right)
    }
    {
      \tensorq{X}
    }
    =
    \tensordot{
      \left(
        \left. \pd{\tensorf{f} \left( \tensorq{G} \right)}{\tensorq{G}} \right|_{\tensorq{G} = \ln \generictensor}
        \left[
          r_{p+s} \left( \ad_{\ln \generictensor}\right)
          \left[
            \tensorq{X}
          \right]
        \right]
        \right)
    }
    {
      \tensorq{X}
    }
    ,
  \end{equation}
  where we have denoted
  \begin{equation}
    \label{eq:14}
    r_q(x)
    =_{\bydefinition}
    \frac{\cosh \left( \frac{q}{2} x \right)}{\sinh \left( \frac{1}{2} x \right)} x.
  \end{equation}
  The function $r_q(x)$ is an even function, $r_q(x) = r_q(-x)$, and it is positive for any parameter value $q \in \R$. Thanks to the positivity of $r_q$ the operator $\sqrt{r_{p+s} \left( \ad_{\ln \generictensor} \right)}$ is well-defined, and, furthermore, the operator is invertible, and thus it is a bijection. (See the same discussion following the equation~\eqref{eq:65}.) Interpreting the operator $\sqrt{r_{p+s} \left( \ad_{\ln \generictensor} \right)}$ via the corresponding formal power series, and exploiting Lemma~\ref{lm:5} on commutativity of derivative and the commutator operator, we see that~\eqref{eq:71} can be further manipulated as
  \begin{multline}
    \label{eq:48}
    \tensordot{
      \left(
        \left. \pd{\tensorf{f} \left( \tensorq{G} \right)}{\tensorq{G}} \right|_{\tensorq{G} = \ln \generictensor}
        \left[
          r_{p+s} \left( \ad_{\ln \generictensor}\right)
          \left[
            \tensorq{X}
          \right]
        \right]
        \right)
    }
    {
      \tensorq{X}
    }
    =
    \tensordot{
      \left(
        \left. \pd{\tensorf{f} \left( \tensorq{G} \right)}{\tensorq{G}} \right|_{\tensorq{G} = \ln \generictensor}
        \left[
          \sqrt{r_{p+s} \left( \ad_{\ln \generictensor}\right)}
          \left[
            \sqrt{r_{p+s} \left( \ad_{\ln \generictensor}\right)}
            \left[
              \tensorq{X}
            \right]
          \right]
        \right]
      \right)
    }
    {
      \tensorq{X}
    }
    \\
    =
    \tensordot{
      \left(
        \sqrt{r_{p+s} \left( \ad_{\ln \generictensor}\right)}
        \left[
          \left. \pd{\tensorf{f} \left( \tensorq{G} \right)}{\tensorq{G}} \right|_{\tensorq{G} = \ln \generictensor}
          \left[
            \sqrt{r_{p+s} \left( \ad_{\ln \generictensor}\right)}
            \left[
              \tensorq{X}
            \right]
          \right]
        \right]
      \right)
    }
    {
      \tensorq{X}
    }
    \\
    =
    \tensordot{
      \left(
        \left. \pd{\tensorf{f} \left( \tensorq{G} \right)}{\tensorq{G}} \right|_{\tensorq{G} = \ln \generictensor}
        \left[
            \sqrt{r_{p+s} \left( \ad_{\ln \generictensor}\right)}
            \left[
              \tensorq{X}
            \right]
        \right]
      \right)
    }
    {
      \left(
        \sqrt{r_{p+s} \left( \ad_{\ln \generictensor}\right)}
        \left[
          \tensorq{X}
        \right]
      \right)
    }
    =
    \tensordot{
      \left(
        \pd{\tensorf{f} \left( \tensorq{G} \right)}{\tensorq{G}} 
        \left[
            \sqrt{r_{p+s} \left( \ad_{ \tensorq{G}}\right)}
            \left[
              \tensorq{X}
            \right]
        \right]
      \right)
    }
    {
      \left(
        \sqrt{r_{p+s} \left( \ad_{\tensorq{G}}\right)}
        \left[
          \tensorq{X}
        \right]
      \right)
    }
    ,
  \end{multline}
  where we have used the characterisation of transpose operator, see~\eqref{eq:85} in Lemma~\eqref{lm:6}. Consequently, we have
  \begin{equation}
    \label{eq:79}
    \tensordot{
      \left(
        \pd{\tensorf{f} \left( \ln \generictensor \right)}{\generictensor}
        \left[
          \generictensor^p \tensorq{X} \generictensor^{-s}
          +
          \generictensor^{-s} \tensorq{X} \generictensor^p
        \right]
      \right)
    }
    {
      \tensorq{X}
    }
    =
    \tensordot{
      \left(
        \pd{\tensorf{f} \left( \tensorq{G} \right)}{\tensorq{G}} 
        \left[
            \sqrt{r_{p+s} \left( \ad_{ \tensorq{G}}\right)}
            \left[
              \tensorq{X}
            \right]
        \right]
      \right)
    }
    {
      \left(
        \sqrt{r_{p+s} \left( \ad_{\tensorq{G}}\right)}
        \left[
          \tensorq{X}
        \right]
      \right)
    }
    .
  \end{equation}
  Since $\sqrt{r_{p+s} \left( \ad_{\tensorq{G}}\right)}$ is a bijection, we see that the equality~\eqref{eq:79} implies that~\eqref{eq:72} is equivalent to~\eqref{eq:73}.
\end{proof}

\section{Conclusion}
\label{sec:conclusion}

We have shown that the calculus based on Lemma~\ref{lm:1}--Lemma~\ref{lm:6} can be extremely effective in solving problems in tensor/matrix analysis. However, so far all proofs exploited the power series based functional calculus. In this setting the matrix exponential, the matrix logarithm and other functions of matrices are defined by the corresponding power series, that is
\begin{subequations}
  \label{eq:90}
  \begin{align}
    \label{eq:91}
    \exponential{\generictensor} &=_{\bydefinition} \sum_{n=0}^{+\infty} \frac{\generictensor^n}{n!}, \\
    \label{eq:92}
    \ln \generictensor &=_{\bydefinition} \sum_{n=1}^{+\infty} (-1)^{n+1} \frac{\left( \generictensor - \identity \right)^n}{n}, \\
    \label{eq:94}
    \tensorf{f} \left(\generictensor \right) &=_{\bydefinition} \sum_{n=0}^{+\infty} f_n \generictensor^n,
  \end{align}
\end{subequations}
whenever the real valued function $f$ has the power series expansion $f(x) = \sum_{n=0}^{+\infty} f_nx^n$. These tensor/matrix valued functions are well-defined for \emph{any} tensor/matrix $\generictensor \in \Mat$ provided that the corresponding power series converge. Furthermore, the formulae for the derivative of tensor/matrix valued functions of tensors/matrices are obtained by the differentiation of the corresponding power series, and they have, in many cases, \emph{nice representation} in terms of the $\ad_{\generictensor}$ operator, see, for example~\eqref{eq:28}. However, the requirement on convergence of the power series brings severe restrictions regarding the applicability of the calculus. For example, while the matrix exponential formula~\eqref{eq:91} is applicable for any matrix, the matrix logarithm formula~\eqref{eq:92} is applicable only for matrices $\generictensor$ not to far from the identity matrix, that is if $\absnorm{\generictensor  - \identity} < 1$. The question is whether we can preserve the simplicity of commutator based calculus of Lemma~\ref{lm:1}--Lemma~\ref{lm:6}, and yet get beyond the power series representation. If we restrict ourselves to a nicer class of tensors/matrices, then the answer is yes.

In particular, if we work with \emph{symmetric tensors/matrices}, we can work with the spectral decomposition based calculus, that is for $\generictensor \in \SymMat$ we define
\begin{equation}
  \label{eq:93}
  \tensorf{f} \left( \generictensor \right)
  =_{\bydefinition}
  \sum_{i=1}^d f(\lambda_i) \tensortensor{\vec{v}_i}{\vec{v}_i},
\end{equation}
where $\generictensor = \sum_{i=1}^d \lambda_i \tensortensor{\vec{v}_i}{\vec{v}_i}$ denotes the spectral decomposition of the given tensor/matrix. (Here $\left\{ \lambda_i \right\}_{i=1}^d$ denote the eigenvalues of $\generictensor$ and  $\left\{ \vec{v}_i \right\}_{i=1}^d$ denote the orthonormal eigenvectors of $\generictensor$.) Clearly, if the logarithm is defined via~\eqref{eq:93}, we see that for matrices with positive eigenvalues we can go much farther beyond the condition $\absnorm{\generictensor  - \identity} < 1$ necessary for the definition~\eqref{eq:92}, and a similar observation holds for other tensor-matrix valued functions defined by power series expansions. This is good news. On the other hand, the derivatives of tensor/matrix valued functions of tensors/matrices have, in this setting, an explicit representation via the Daleckii--Krein formula, see \cite{daletskii.jl.krein.sg:integration} or \cite{bhatia.r:matrix}, but this \emph{representation is inconvenient} in symbolic manipulations. The formulae based on the commutator $\ad_{\generictensor}$ are much nicer. And here comes the point.

If we restrict ourselves to symmetric tensors/matrices and tensor/matrix valued functions operating on the symmetric tensors/matrices, then we can have the best from both worlds. We can work with the commutator based calculus and all its nice algebraic properties, and yet all the operators can be defined without the underlying formal power series representation, that is without the restrictions imposed by the convergence of the corresponding series.

To achieve this we must define the operator associated to the symbol
\begin{equation}
  \label{eq:95}
  f(\ad_{\generictensor})
\end{equation}
in such a way that the operator $f(\ad_{\generictensor})$ is \emph{applicable to any symmetric tensor/matrix}, and that the operator $f(\ad_{\generictensor})$ \emph{coincides}, for the symmetric tensors/matrices $\tensorq{A}$, \emph{with the operator $f(\ad_{\generictensor})$ defined via the corresponding formal power series}, as long as the power series converges. Furthermore, the operator must inherit all the algebraic properties listed in Lemma~\ref{lm:1}--Lemma~\ref{lm:6}. This can be done using the following definition.

\begin{definition}[Operator $f(\ad_{\tensorq{G}})$ for  $\tensorq{G} \in \SymMat$]
  \label{dfn:1}
  Let $\tensorq{G} \in \SymMat$ be a given symmetric matrix, and let $\left\{ g_i \right\}_{i=1}^d \subset \R$ be its eigenvalues. The matrix $\tensorq{G}$ can be diagonalised by an orthogonal similarity transformation, and we denote the corresponding diagonal matrix by $\tensorq{g}=_{\bydefinition} \diag(g_i)_{i=1}^d$, while the corresponding similarity transformation matrix is denoted by~$\tensorq{Q}$, that is
  \begin{equation}
    \label{eq:96}
    \tensorq{g} = \transpose{\tensorq{Q}} \tensorq{G} \tensorq{Q},
  \end{equation}
  with $\transpose{\tensorq{Q}} = \inverse{\tensorq{Q}}$. For any function $f: \R \to \R$, we define the matrix $f(\times\tensorq{g})$ via its components in the basis wherein the matrix~$\tensorq{G}$ has the diagonal form~$\tensorq{g}$ as
  \begin{equation}
    \label{eq:97}
    \tensor{\left[ f \left( \times\tensorq{g} \right) \right]}{_{ij}} =_\bydefinition f(g_i-g_j), \quad i,j \in\{1,\ldots,d\}.
  \end{equation}
  Furthermore, the linear operator $f(\ad_{\tensorq{G}}): \Mat \to \Mat$ for $\tensorq{X} \in \Mat$ is defined as \begin{equation}\label{had}
    f(\ad_{\tensorq{G}}) \left[ \tensorq{X} \right]
    =_{\bydefinition}
    \tensorq{Q}
    \left(
      \tensorschur{f(\times\tensorq{g})}{ \transpose{\tensorq{Q}} \tensorq{X} \tensorq{Q}}
    \right)
    \transpose{\tensorq{Q}},
  \end{equation}
  where $\tensorschur{}{}$ denotes the Schur/Hadamard product in the basis where the matrix $\tensorq{G}$ has the diagonal form~$\tensorq{g}$.
\end{definition}

In the follow-up work, see~\cite{bathory-calculus}, we rigorously show that the operator $f(\ad_{\tensorq{G}})$ introduced in Definition~\ref{dfn:1} has all the required properties, and that the \emph{commutator based functional calculus} on symmetric tensors/matrices is well-defined. 

%\todo[inline]{Go carefully through the text and check what holds for symmetric/general matrices, for example, explain that $\generictensor \tensorq{X} + \tensorq{X} \generictensor$ in~\eqref{eq:13} must be symmetric, otherwise we won't be able to define the derivative---the direction must be a symmetric matrix if we want to work with spectral calculus for $\ln \generictensor$ and so forth. Someone should do this for infinite dimensional operators---spectral measure :-)}

% Using the functional calculus we can also show the equivalence between the standard representation formula~\eqref{eq:4} and the new representation formula~\eqref{eq:7} by direct substitution, see Section~TODO.
 
%%% Local Variables:
%%% mode: LaTeX
%%% TeX-master: "../logarithmic-corotational-derivative-note"
%%% End:

\bibliographystyle{chicago}
\bibliography{vit-prusa,bathory-calculus}

\appendix

\section{Classical Lie algebra formulae from the perspective of power series based calculus}
\label{sec:lie-algebra-formulae}
For the sake of completeness we present a derivation of some well known formulae for the commutator operator and the matrix exponential we have used without proof in the main text body. All arguments are based on the power series calculus. First, we derive a representation formula for the powers of the commutator operator. This representation formula and its proof are known, see, for example, \cite[Exercise 14, Chapter 3]{hall.b:lie}.

\begin{lemma}[Powers of commutator operator]
  \label{lm:commutator-power}
  Let $\generictensor \in \Mat$ and $\tensorq{X} \in \Mat$ be arbitrary matrices and let $m \in \N$. Then
  \begin{equation}
    \label{eq:appendix-1}
    \ad_{\tensorq A}^m \left[ \tensorq{X} \right]
    =
    \sum_{k=0}^m
    \binom{m}{k}
    \generictensor^k
    \tensorq{X}
    \left(
      -
      \generictensor
    \right)^{m-k}
    .
  \end{equation}
\end{lemma}
\begin{proof}
  We prove the proposition by induction. The proposition holds for $m=1$. Now we assume that the proposition~\eqref{eq:appendix-1} holds form $m$, and we want to show that it also holds for $m+1$, that is we want to show
  \begin{equation}
    \label{eq:appendix-2}
    \ad_{\tensorq A}^{m+1} \left[ \tensorq{X} \right]
    =
    \sum_{k=0}^{m+1}
    \binom{m+1}{k}
    \generictensor^k
    \tensorq{X}
    \left(
      -
      \generictensor
    \right)^{m+1-k}
    .
  \end{equation}
  This follows form a simple manipulation based on the binomial coefficients identity $\binom{m+1}{l} = \binom{m}{l-1} + \binom{m}{l}$. We have
  \begin{multline}
    \label{eq:appendix-3}
    \ad_{\tensorq A}^{m+1} \left[ \tensorq{X} \right]
    =
    \ad_{\tensorq A} \left[\ad_{\tensorq A}^m  \left[ \tensorq{X} \right] \right]
    =
    \ad_{\tensorq A}
    \left[
      \sum_{k=0}^m
      \binom{m}{k}
      \generictensor^k
      \tensorq{X}
      \left(
        -
        \generictensor
      \right)^{m-k}
    \right]
    =
      \sum_{k=0}^m
      \binom{m}{k}
      \generictensor^{k+1}
      \tensorq{X}
      \left(
        -
        \generictensor
      \right)^{m-k}
      +
      \sum_{k=0}^m
      \binom{m}{k}
      \generictensor^{k}
      \tensorq{X}
      \left(
        -
        \generictensor
      \right)^{m+1-k}
      \\
      =
      \sum_{l=1}^{m+1}
      \binom{m}{l-1}
      \generictensor^{l}
      \tensorq{X}
      \left(
        -
        \generictensor
      \right)^{m+1-l}
      +
      \sum_{l=0}^{m}
      \binom{m}{l}
      \generictensor^{l}
      \tensorq{X}
      \left(
        -
        \generictensor
      \right)^{m+1-l}
      \\
      =
      \binom{m}{0}
      \generictensor^{0}
      \tensorq{X}
      \left(
        -
        \generictensor
      \right)^{m+1}
      +
      \sum_{l=1}^{m}
      \left(
        \binom{m}{l-1}
        +
        \binom{m}{l}
        \right)
      \generictensor^{l}
      \tensorq{X}
      \left(
        -
        \generictensor
      \right)^{m+1-l}
      +
      \binom{m}{m}
      \generictensor^{m+1}
      \tensorq{X}
      \left(
        -
        \generictensor
      \right)^{0}
      =
      \sum_{l=0}^{m+1}
      \binom{m+1}{l}
      \generictensor^l
      \tensorq{X}
      \left(
        -
        \generictensor
      \right)^{m+1-l}
      .
    \end{multline}
\end{proof}
The formula for the powers of commutator operator immediately gives us the Campbell lemma, see, for example, \cite[Proposition 3.35]{hall.b:lie}. 
\begin{lemma}[Campbell formula, exponential of the commutator operator]
  \label{lm:campbell}
  Let $\generictensor \in \Mat$ and $\tensorq{X} \in \Mat$ be arbitrary matrices, then
  \begin{equation}
    \label{eq:appendix-5}
    \exponential{\ad_{\generictensor}}
    \left[
      \tensorq{X}
    \right]
    =
    \exponential{\generictensor} \tensorq{X} \exponential{-\generictensor}    
    .
  \end{equation}
\end{lemma}
\begin{proof}
  The standard proof is based on direct calculation with the matrix exponential power series. Using the Cauchy product for power series we first find that
  \begin{equation}
    \label{eq:appendix-6}
    \exponential{\generictensor} \tensorq{X} \exponential{-\generictensor}
    =
    \left(\sum_{k=0}^{+\infty} \frac{ \generictensor^k }{k!} \right)
    \tensorq{X}
    \left(\sum_{l=0}^{+\infty} \frac{ \left( - \generictensor \right)^l }{l!} \right)
    =
    \sum_{j=0}^{+\infty}
    \sum_{l=0}^{j}
    \frac{\generictensor^l }{l!}
    \tensorq{X}
    \frac{ \left(-\generictensor \right)^{j-l} }{\left( j-l \right)!}
    .
  \end{equation}
  On the other hand, the formula for exponential of the commutator operator gives us
  \begin{equation}
    \label{eq:appendix-7}
    \exponential{\ad_{\generictensor}}
    \left[
      \tensorq{X}
    \right]
    =
    \sum_{m=0}^{+\infty} \frac{\ad_{\generictensor}^m \left[ \tensorq{X} \right]}{m!}
    =
    \sum_{m=0}^{+\infty}
    \frac{1}{m!}
    \sum_{k=0}^m
    \binom{m}{k}
    \generictensor^k
    \tensorq{X}
    \left(
      -
      \generictensor
    \right)^{m-k}
    =
    \sum_{m=0}^{+\infty}
    \sum_{k=0}^m
    \frac{\generictensor^k}{k!}
    \tensorq{X}
    \frac{
      \left(
        -
        \generictensor
      \right)^{m-k}
    }
    {
      \left(m-k\right)!
    }
    ,
  \end{equation}
  where the second equality in~\eqref{eq:appendix-7} follows from Lemma~\ref{lm:commutator-power}.
\end{proof}

Using Lemma~\ref{lm:campbell} it is straightforward to find a commutator based formula for the derivative of the matrix exponential, see, for example, \cite[Theorem 5.4]{hall.b:lie}. Note however that our the proof of Lemma~\ref{lm:derivative-of-matrix-exponential} is different from the standard one, and it is based on the commutator calculus.

\begin{lemma}[Derivative of matrix exponential]
  \label{lm:derivative-of-matrix-exponential}
  Let $\generictensor \in \Mat$ and  $\tensorq{X} \in \Mat$ be arbitrary given matrices, then
  \begin{equation}
    \label{eq:appendix-4}
    \pd{\exponential{\generictensor}}{\generictensor} \left[ \tensorq{X} \right]
    =
    \exponential{\generictensor} \frac{1 - \exponential{-\ad_{\generictensor}}}{\ad_{\generictensor}} [\tensorq{X}].
  \end{equation}
\end{lemma}

\begin{proof}
  Our objective is to find an operator equation that is satisfied by the derivative, and then to solve this equation. We first differentiate the identity matrix $\identity = \exponential{\generictensor - \generictensor}$ using the product rule, and we get
  \begin{equation}
    \label{eq:appendix-8}
    \tensorzero
    =
    \pd{}{\generictensor}
    \left(
       \exponential{\generictensor - \generictensor}
    \right)
    \left[
      \tensorq{X}
    \right]
    =
    \pd{}{\generictensor}
    \left(
       \exponential{\generictensor}\exponential{- \generictensor}
    \right)
    \left[
      \tensorq{X}
    \right]
    =
    \pd{\exponential{\generictensor}}{\generictensor}
    \left[
      \tensorq{X}
    \right]
    \exponential{- \generictensor}
    +
    \exponential{\generictensor}
    \pd{\exponential{- \generictensor}}{\generictensor}
    \left[
      \tensorq{X}
    \right]
    .
  \end{equation}
  This manipulation reveals that
  \begin{equation}
    \label{eq:appendix-9}
    \pd{\exponential{- \generictensor}}{\generictensor}
    \left[
      \tensorq{X}
    \right]
    =
    -
    \exponential{- \generictensor}
    \pd{\exponential{\generictensor}}{\generictensor}
    \left[
      \tensorq{X}
    \right]
    \exponential{- \generictensor}
    .
  \end{equation}
  This is not surprising since this is in fact a formula that follows from the derivative of matrix inverse and the chain rule.
  
  Now we differentiate the identity 
  \begin{equation}
    \label{eq:appendix-10}
    \exponential{- \generictensor}
    \generictensor
    \exponential{\generictensor}
    =
    \generictensor
  \end{equation}
  using the product rule. (Recall that the matrices $\generictensor$, $\exponential{\generictensor}$ and $\exponential{-\generictensor}$ commute.) The differentiation of~\eqref{eq:appendix-10} in the direction $\tensorq{X}$ yields
  \begin{equation}
    \label{eq:appendix-11}
    \pd{\exponential{-\generictensor}}{\generictensor}
    \left[
      \tensorq{X}
    \right]
    \generictensor
    \exponential{\generictensor}
    +
    \exponential{- \generictensor}
    \tensorq{X}
    \exponential{\generictensor}
    +
    \exponential{- \generictensor}
    \generictensor
    \pd{\exponential{\generictensor}}{\generictensor}
    \left[
      \tensorq{X}
    \right]
    =
    \tensorq{X}
    .
  \end{equation}
  We use the identity for the derivative of the exponential $\exponential{-\generictensor}$, see~\eqref{eq:appendix-9}, and we regroup the terms in~\eqref{eq:appendix-11} as
  \begin{equation}
    \label{eq:appendix-12}
    -
    \exponential{- \generictensor}
    \pd{\exponential{\generictensor}}{\generictensor}
    \left[
      \tensorq{X}
    \right]
    \generictensor
    +
    \generictensor
    \exponential{- \generictensor}
    \pd{\exponential{\generictensor}}{\generictensor}
    \left[
      \tensorq{X}
    \right]
    +
    \exponential{- \generictensor}
    \tensorq{X}
    \exponential{\generictensor}
    =
    \tensorq{X}.
  \end{equation}
  The first two terms on the left-hand side give us the commutator between $\generictensor$ and
  $
  \exponential{- \generictensor}
  \pd{\exponential{\generictensor}}{\generictensor}
  \left[
    \tensorq{X}
  \right]$,
  while the third term on the left-hand side can be rewritten using Lemma~\ref{lm:campbell}. This yields
  \begin{equation}
    \label{eq:appendix-13}
    \ad_{\generictensor}
    \left[
      \exponential{- \generictensor}
      \pd{\exponential{\generictensor}}{\generictensor}
      \left[
        \tensorq{X}
      \right]
    \right]
    +
    \exponential{-\ad_\generictensor} \left[ \tensorq{X} \right]
    =
    \tensorq{X},
  \end{equation}
  and we can formally manipulate this equation as 
  \begin{equation}
    \label{eq:appendix-14}
      \exponential{- \generictensor}
      \pd{\exponential{\generictensor}}{\generictensor}
      \left[
        \tensorq{X}
      \right]
      =
      \frac{
        1
        -
        \exponential{-\ad_\generictensor}
      }
      {
        \ad_{\generictensor}
      }
      \left[
        \tensorq{X}
      \right],
    \end{equation}
    which gives us the proposition~\eqref{eq:appendix-4} in Lemma~\ref{lm:derivative-of-matrix-exponential}.
\end{proof}

%%% Local Variables:
%%% mode: LaTeX
%%% TeX-master: "../logarithmic-corotational-derivative-note"
%%% End:

%\addtocontents{toc}{\protect\end{multicols}} % workaround for table of contents in two columns in amsart documentclass
\end{document}